\documentclass[a4paper,10pt]{article}
\usepackage{amsmath,amsfonts,amssymb}
\usepackage{amsthm}
\usepackage{mathtools}
\usepackage{amsopn}
\usepackage{tabularx,lipsum,environ}
\usepackage{paralist}
\usepackage{microtype}
\usepackage{authblk}
\usepackage{hyperref}
\usepackage{fancyhdr}
\usepackage{rotating}
\usepackage{overpic}
\usepackage{ucs}
\usepackage{enumerate}
\usepackage{graphicx}
\usepackage{booktabs}
\usepackage{varwidth}
\usepackage{subfig}

\newtheorem{theorem}{Theorem}
\newtheorem{proposition}{Proposition}
\newtheorem{corollary}{Corollary}
\newtheorem{definition}{Definition}

\usepackage{graphicx}
\usepackage{color}
\usepackage{comment}
\usepackage[nobreak]{cite} 
\newcommand{\revised}[1]{\textcolor{black}{{#1}}}
\newcommand{\retouched}[1]{\textcolor{black}{{#1}}}
  
\newcommand{\tw}{{\rm tw}}  
\newcommand{\pw}{{\rm pw}}  
\newcommand{\vc}{{\rm vc}}  
\newcommand{\ut}{{\rm u}}   
\newcommand{\PP}{{\mathcal P}}  

\title{Computational Complexity of Hedonic Games \\on Sparse Graphs\thanks{This work was partially supported by JSPS KAKENHI Grant Numbers JP17K19960, 17H01698, 19K21537.}}
\author[1]{Tesshu Hanaka}
\author[2]{Hironori Kiya}
\author[2]{Yasuhide Maei}
\author[2]{Hirotaka Ono}

\affil[1]{Chuo University, Tokyo, Japan}
\affil[2]{Nagoya University, Nagoya, Japan}

\date{}
\begin{document}
%



%
\maketitle              
\begin{abstract}

The additively separable hedonic game (ASHG) is a model of coalition formation games on graphs. 
In this paper, we intensively and extensively investigate the computational complexity of finding several desirable solutions, such as a Nash stable solution, a maximum utilitarian solution, and a maximum egalitarian solution in ASHGs on sparse graphs including bounded-degree graphs, bounded-treewidth graphs, and near-planar graphs. For example, we show that finding a maximum egalitarian solution is weakly NP-hard even on graphs of treewidth 2, whereas it can be solvable in polynomial time on trees. Moreover, we give a pseudo fixed parameter algorithm when parameterized by treewidth.

\end{abstract}

\section{Introduction}
In this paper, we investigate the computational complexity of additively separable hedonic games on sparse graphs from the viewpoint of several solution concepts. 

Given the set of agents, the coalition formation game is a model of finding a partition of the set of agents into subsets under a certain criterion, where each of the subsets is called a \emph{coalition}. Such a partition is called a \emph{coalition structure}. The {\emph{hedonic game}} is a variant of coalition formation games, where each agent has the utility associated with his/her joining coalition. In the typical setting, if an agent belongs to a coalition where his/her favorite agents also belong to, his/her utility is high and he/she feels comfortable. Contrarily, if he/she does not like many members in the coalition, his/her utility must be low; since he/she feels uncomfortable, he/she would like to move to another coalition. Although the model of hedonic games is very simple, it is useful to represent many practical situations, such as formation of research team~\cite{RePEc:eee:mateco:v:40:y:2004:i:8:p:869-887}, formation of coalition government~\cite{le2008gamson}, clustering in social networks~\cite {Aziz:2014:FHG:2615731.2615736,McSweeney:2012:GTF:2456719.2457005,olsen2009nash}, multi-agent distributed task assignment~\cite{saad2010hedonic}, and so on.

The  {\em additively separable hedonic game} (ASHG) is a class of hedonic games, where the utility forms an additively separable function. In ASHG, an agent has a certain \emph{valuation} for each of the agents, which represents his/her preference. The valuation could be positive, negative or $0$. If the valuation of agent $u$ for agent $v$ is positive, agent $u$ prefers agent $v$, and if it is negative, agent $u$ does not prefer agent $v$. If it is $0$, agent $u$ has no interest for agent $v$.  The utility of agent $u$ for $u$'s joining coalition $C$ is defined by the sum of valuations of agent $u$ for other agents in $C$. This setting is considered not very but reasonably general. Due to this definition, it can be also defined by an edge-weighted directed graph, where the weight of edge $(u,v)$ represents the valuation of $u$ to $v$. If a valuation is $0$, we can remove the corresponding edge. Note that the undirected setting is possible, and in the case the valuations are symmetric; the valuation of agent $u$ for agent $v$ is always equal to the one of agent $v$ for agent $u$. 

In the study of hedonic games, several solution concepts are considered important and well investigated. 
One of the most natural solution concepts is \emph{maximum utilitarian}, which is so-called a global optimal solution; it is a coalition structure that maximizes the total sum of the utilities of all the agents. The total sum of the utilities is also called \emph{social welfare}.
Another concept of a global optimal solution is \emph{maximum egalitarian}. It maximizes the minimum utility of an agent among all the agents. That is, it makes the unhappiest agent as happy as possible. 
Nash-stability, envy-free and max envy-free are more personalized concepts of the solutions. 
A coalition structure is called \emph{Nash-stable} if no agent has an incentive to move to another coalition from the current joining coalition. Such an incentive to move to another coalition is also called a \emph{deviation}. Agent $u$ feels envious of $v$ if $u$ can increase his/her utility by exchanging the coalitions of $u$ and $v$. A coalition structure is \emph{envy-free} if any agent does not envy any other agent. Furthermore, the best one among the envy-free coalition structures is also meaningful; it is an envy-free coalition structure with maximum social welfare. 
Some other concepts are also considered, though we focus on these concepts in this paper. 
%

Of course, it is not trivial to find a coalition structure satisfying above mentioned solution concepts. 
Ballester studies the computational complexity for finding coalition structures of several concepts including the above mentioned ones~\cite{ballester2004np}. More precisely, he shows that determining whether there is a Nash stable, an individually stable, and a core stable coalition structure is NP-complete.
In~\cite{sung2010computational}, Sung and Dimitrov show that the same results hold for ASHG.
Aziz et al. investigate the computational complexity for many concepts including the above five solution concepts~\cite{Aziz2013}.
In summary, ASHG is unfortunately NP-hard for the above five solution concepts. 
These hardness results are however proven without any assumption about graph structures. For example, some of the proofs suppose that graphs are weighted complete graphs. This might be a problem, because graphs appearing in ASHGs for practical applications are so-called social networks; they are far from weighted complete graphs and known to be rather sparse or tree-like~\cite{DUNBAR1992469,adcock2013tree}.  
What if we restrict the input graphs of  ASHG to sparse graphs? This is the motivation of this research.  


In this paper, we investigate the computational complexity of ASHG on sparse graphs from the above five solution concepts. The sparsity that we consider in this paper is as follows: graphs with bounded degree,  
graphs with bounded treewidth and near-planar graphs. The degree is a very natural parameter that characterizes the sparsity of graphs. In social networks, the degree represents the number of friends, which is usually much smaller than the size of network. The {\emph{treewidth}} is a parameter that represents how tree-like a graph is. As Adcock, Sullivan and Mahoney pointed out in ~\cite{adcock2013tree}, many large social and information networks have tree-like structures, which implies the significance to investigate the computational complexity of ASHG on graphs with bounded treewidth. Near-planar graphs here are $p$-apex graphs. 
A graph $G$ is said to be $p$-apex if $G$ becomes planar after deleting $p$ vertices or fewer vertices.
Near-planarity is less important than the former two in the context of social networks, though it also has many practical applications such as transportation networks. Note that all of these sparsity concepts are represented by parameters, i.e., treewidth, maximum degree and $p$-apex. In that sense, we consider the parameterized complexity of ASHG of several solution concepts in this paper.  

This is not the first work that focuses on the parameterized complexity of ASHG. Peters presents that Nash-stable, Maximum Utilitarian, Maximum Egalitarian and Envy-free coalition structures can be computed in $2^{\tw \Delta^2}n^{O(1)}$ time, where $\tw$ is the treewidth and $\Delta$ is the maximum degree of an input graph~\cite{peters2016graphical}. In other word, it is {\em fixed parameter tractable (FPT)} with respect to treewidth and maximum degree. This implies that if both of the treewidth and the maximum degree are small, we can efficiently find desirable coalition structures. This result raises the following natural question: 
is finding these desirable coalition structures still FPT when parameterized by either the treewidth or the maximum degree?

This paper answers the question from various viewpoints. Different from the case parameterized by treewidth and maximum degree, the time complexity varies depending on the solution concepts. For example, we can compute a maximum utilitarian coalition structure in  ${\tw}^{O(\tw)}n$ time, whereas computing a maximum egalitarian coalition structure is weakly NP-hard even for graphs with treewidth at most $2$. Some other results of ours are summarized in Table \ref{tb:matome}. For more details, see Section \ref{sec:1.contribution}. Also some related results are summarized in Section \ref{sec:1.related}. 

\subsection{Our contribution}\label{sec:1.contribution} 
\begin{table}[tbp]
  \centering
    \caption{Complexity of ASHGs}
  \begin{tabular}{|l|l|c|}
  \hline
Concept & Time complexity to compute &Reference \\ \hline
Nash stable & NP-hard   &~\cite{sung2010computational} \\ 
&  PLS-complete (symm) &~\cite{gairing2010computing} \\ 
 & \textbf{PLS-complete (symm, $\Delta = 7$)} & [Th.\ref{thm:md7}] \\
 & \textbf{${\tw}^{O(\tw)}n$ (symm, FPT by treewidth)} & [Cor.\ref{cor:Nash-degree2}] \\ \hline
Max Utilitarian & strongly NP-hard (symm) &~\cite{Aziz2013} \\ 
 & \textbf{strongly NP-hard (symm, 3-apex)} & [Th.\ref{thm:mu_3apex}] \\ 
 & \textbf{${\tw}^{O(\tw)}n$ (FPT by treewidth)} & [Th.\ref{thm:treewidth:hedonic}] \\ \hline
Max Egalitarian & strongly NP-hard &~\cite{Aziz2013} \\ 
 & \textbf{weakly NP-hard (symm, 2-apex, $\vc= 4$)} & [Th.\ref{thm:vc:Egal}] \\ 
 & \textbf{weakly NP-hard (symm, planar, $\pw = 4$, $\tw = 2$)} & [Th.\ref{thm:tw:Egal}] \\ 
 & \textbf{strongly NP-hard (symm)} & [Th.\ref{thm:Envy-Egal:strongNP}] \\ 
 & \textbf{linear (symm, tree)} & [Th.\ref{thm:tree:Egal}] \\ 
  & \textbf{P (tree)} & [Th.\ref{thm:asym:tree}] \\ 
 & \textbf{$({\tw}W)^{O(\tw)}n$     (pseudo FPT by treewidth)} & [Th.\ref{thm:pseudo_tw:Egal}] \\ \hline
Envy-free &  trivial  &~\cite{Aziz2013} \\  \hline
Max Envy-free  & \textbf{weakly NP-hard (symm, planar, $\vc= 2$, $\tw = 2$)} & [Th.\ref{thm:vc:Envy-free}] \\ 
 & \textbf{strongly NP-hard (symm)} & [Th.\ref{thm:Envy-Egal:strongNP}] \\
 & \textbf{linear  (symm, tree)} & [Th.\ref{thm:tree:Egal}] \\ \hline
  \end{tabular}
\end{table}
\label{tb:matome}


We first study (symmetric) \textsc{Nash stable} on bounded degree graphs.
We show that the problem is PLS-complete even on graphs with maximum degree $7$. 
PLS is a complexity class of a pair of an optimization problem and a local search for it. 
It is originally introduced to capture the difficulty of finding a locally optimal solution of an optimization problem. In the context of hedonic games, a deviation corresponds to an improvement in local search, and thus PLS or PLS-completeness is also used to model the difficulty of finding a stable solution.

We next show that \textsc{Max Utilitarian} is strongly NP-hard on 3-apex graphs, whereas it can be solved in time ${\tw}^{O(\tw)}n$, and hence it is FPT when parameterized by treewidth $\tw$.
For \textsc{Max Envy-free} , we show that  the problem is weakly NP-hard on series-parallel graphs with vertex cover number at most $2$ whereas finding an envy-free partition is trivial~\cite{Aziz2013}. 

Finally, we investigate the computational complexity of \textsc{Max Egalitarian}.
We show that \textsc{Max Egalitarian} is weakly NP-hard on 2-apex graphs with vertex cover number at most $4$ and planer graphs with pathwidth at most $4$ and treewidth at most $2$. 
Moreover, we show that \textsc{Max Egalitarian} and \textsc{Max Envy-free} are strongly NP-hard even if the preferences are symmetric.
In contrast, an egalitarian and envy-free partition with maximum social welfare can be found in linear time on trees if the preferences are symmetric. Moreover, \textsc{Max Egalitarian} can be computed in polynomial time even if the preferences are asymmetric.
In the end of this paper, we give a pseudo FPT algorithm when parameterized by treewidth.

\subsection{Related work}\label{sec:1.related}
The {\em coalition formation game} is first introduced by Dreze and Greenber~\cite{dreze1980hedonic} in the  field of Economics. Based on the concept of coalition formation games, Banerjee, Konishi and S\"onmez~\cite{banerjee2001core} and Bogomolnaia and Jackson~\cite{bogomolnaia2002stability} study some  stability and core concepts on hedonic games.
For the computational complexity on hedonic games, Ballester shows that finding several coalition structures including Nash stable, core stable, and individually stable coalition structures is NP-complete~\cite{ballester2004np}. 
For ASHGs, Aziz et al. investigate the computational complexity of finding several desirable coalition structures~\cite{Aziz2013}.
 Gairing and Savani~\cite{gairing2010computing} show that computing a Nash stable coalition structure is PLS-complete in symmetric AGHGs whereas Bogomolnaia and Jackson~\cite{bogomolnaia2002stability} prove that a Nash stable coalition structure always exists.
In~\cite{peters2016graphical}, Peters designs parameterized algorithms for computing some coalition structures on hedonic games with respect to treewidth and maximum degree.
\section{Preliminaries}
In this paper, we use the standard graph notations.
For $G=(V,E)$, we define $n=|V|$ and $m=|E|$. 
For $V'\subseteq V$, we denote by $G[V^\prime ]$ the subgraph of $G$ induced by $V^\prime $. 
We denote the closed neighbourhood and the open neighbourhood of a vertex $v$ by $N[v]$ and $N(v)$, respectively.
The degree of $v$ is denoted by $d(v)$. 
Moreover, the maximum degree of $G$ is denoted by $\Delta(G)$. For simplicity, we sometimes omit the subscript $G$.

\subsection{Hedonic game}
An {\em additively separable hedonic game} (ASHG) is defined on a directed edge-weighted graph $G=(V,E,w)$. 
Each vertex $v \in V$ is called an {\em agent}. 
The weight of an edge $e=(u,v)$, denoted by $w_e$ or $w_{uv}$, represents the valuation of $u$ to $v$. 
An ASHG is said to be {\em symmetric} if $w_{uv}=w_{vu}$ holds for any pair of $u$ and $v$.
Any symmetric ASHG can be defined on an {\em undirected} edge-weighted graph.
We denote an undirected edge by $\{u,v\}$.
Note that any edge of weight $0$ is removed from a graph.

Let  $\PP$ be a partition of $V$. Then $C\in \PP$ is called a {\em coalition}. We denote by $C_u\in \PP$ the coalition to which an agent  $u \in V$ belongs under $\PP$, and by $E(C_u)$ the set of edges $\{(u, v)\cup (v,u) \in E \mid v\in C_u\}$. 
In ASHGs, the utility of an agent $u$ under $\PP$ is defined as  ${\ut}_{\PP}(u)=\sum_{v\in N(u)\cap C_u} w_{uv}$, which is the sum of weights of edges from $u$ to other agents in the same coalition. 
Also, the {\em social welfare} of $\PP$ is defined as the sum of utilities of all agents under $\PP$.
Note that the social welfare equals to exactly twice the sum of weights of edges in coalitions.

Next, we define several concepts of desirable solution in ASHGs.
\begin{definition}[Nash-stable]
A partition $\PP$ is Nash-stable if there exists no agent $u$ and coalition $C'\neq C_u$ containing $u$, possibly empty, such that
$$\sum_{v\in N(u)\cap C_u} w_{uv} < \sum_{v\in N(u)\cap C'} w_{uv}.$$
\end{definition}


As an important fact, in any symmetric ASHG, a partition with maximum social welfare is Nash-stable by using the potential function argument~\cite{bogomolnaia2002stability}. 
\begin{proposition}\label{prop:stability_max}
In any symmetric ASHG, a partition with maximum social welfare is Nash-stable.
\end{proposition}
Thus, if we can compute a partition with maximum social welfare in a symmetric ASHG, then we also obtain a Nash-stable partition.

\begin{definition}[Envy-free]
We say an agent $u_1\in C_{u_1}$ {\em envies} $u_2\in C_{u_2}$ if the following inequality holds:
\begin{align*}
\sum_{v\in N(u_1)\cap C_{u_1}} w_{u_1 v} <\sum_{v\in N(u_1)\cap (C_{u_2} \backslash \{u_2\} \cup \{u_1\})} w_{u_1 v}.
\end{align*}
That is, $u_1$ envies $u_2$ if the utility of $u_1$ increases by replacing  $u_2$ by $u_1$.
A partition $\PP$ is {\em envy-free} if any agent does not envy an agent.
\end{definition}

\textsc{Nash-stable}, \textsc{Envy-free}, \textsc{Max Envy-free}, \textsc{Max Utilitarian}, and \textsc{Max Egalitarian} are the following problems: Given a weighted graph $G=(V,E,w)$, find a Nash-stable partition, an envy-free partition, an envy-free partition with maximum social welfare, a maximum utilitarian partition, and a maximum egalitarian partition, respectively.






\subsection{Graph classes}
A \retouched{{\em planar}} graph is a graph that can be drawn on the plane in such a way that its edges intersect only at their endpoints. 
For $p\ge 1$, a {\em $p$-apex graph} is a graph that can be planar by removing $p$ vertices  or fewer vertices from it.
Note that a planar graph is a $p$-apex graph for any $p\ge 1$.
A graph $G$ is called {\em series-parallel} if every 2-connected component of $G$ can be constructed by applying {\em series operation} and {\em parallel operation} compositions recursively: The series operation entails subdividing an edge by a new vertex (replacing an edge by two edges in series). The parallel operation entails replacing an edge by two edges in parallel.
It is well-known that a series-parallel graph is planar, and the class of series-parallel graph is equivalent to graphs with treewidth 2.

\subsection{Graph parameters and parameterized complexity}
For the basic definitions of parameterized complexity, such as the classes FPT and XP, refer to~\cite{Cygan2015}.

\begin{definition}[Tree decomposition]
A {\em tree decomposition} of an undirected graph $G=(V,E)$ is defined as 
a pair $\langle {\cal X}, T\rangle$, where ${\cal
X}=\{X_1,X_2,\ldots, X_N \subseteq V\}$, and $T$ is a tree whose nodes
are labeled by $I \in \{1,2,\ldots,N\}$, such that  
\begin{description}
\item[1.] $\bigcup_{i\in I} X_i =V$,
\item[2.] For all $\{u, v\}\in E$, there exists an $X_i$ such that 
	   $\{u, v\} \subseteq X_i$, 
\item[3.] For all $i, j, k \in I$, if $j$ lies on the
	   path from $i$ to $k$ in $T$, then $X_i \cap X_k \subseteq
	   X_j$.   
\end{description}
\revised{Here, $X_i$ is called a {\em bag}.}
The {\em width} of a tree decomposition is defined as $\min_{i\in I} |X_i|-1$, that is, minimum size of a bag minus one.
Furthermore, the {\em treewidth} of $G$, denoted by $\tw(G)$, is minimum possible width of a tree decomposition of $G$.
A tree decomposition $\langle {\cal X}, T\rangle$ is called a {\em path decomposition} if $T$ is a path. 
The {\em pathwidth} of $G$, denoted by $\pw(G)$, is minimum possible width of a path decomposition of $G$.
\end{definition}

We introduce a special type of tree decomposition, a {\em nice tree decomposition}, introduced by Kloks~\cite{Kloks1994}. It is a special binary tree decomposition which has four types of nodes, named {\em leaf}, {\em introduce vertex}, {\em forget}  and {\em join}. 
In~\cite{Cygan2015,Cygan2011}, Cygan et al. added a fifth type, the {\em introduce edge} node.

\begin{definition}[Nice tree decomposition]
 A tree decomposition  $\langle {\cal X}, T\rangle$ is called a  {\em nice tree decomposition} if it satisfies the following:
\begin{description}
\item[1.] $T$ is rooted at a designated node $r \in {I}$ satisfying $|X_r|=0$, called 
the {\em root node}, 
\item[2.] Each node of the tree $T$ has at most two children, 
\item[3.] Each node in $T$ has one of the following five types:
\begin{itemize}
\item A {\em leaf} node $i$ which has no children and its bag $X_i$ satisfies $|X_i| = 0$, 
\item An {\em introduce vertex} node $i$  has one child $j$ with $X_i = X_j \cup \{v\} $ for a vertex $v\in V$, 
\item An {\em introduce edge} node $i$ has one child $j$ and  labeled with an edge $(u, v) \in E$ where $u, v \in X_i = X_j$, 
\item A {\em forget} node $i$ has one child $j$ and satisfies $X_i = X_j \setminus \{v\}$ for
a vertex $v \in V$, 
\item A {\em join} node $i$ has two children nodes $j_1, j_2$ and satisfies $X_i = X_{j_1} = X_{j_2}$.
\end{itemize} 
\end{description}
\end{definition}
Any tree decomposition of width $\omega$ can be transformed into a nice tree decomposition of  $\omega$ with $O(n)$ nodes in linear time~\cite{Cygan2015}.

A {\em vertex cover} $S$ is the set of vertices such that every edge has at least one vertex in $S$. 
The size of minimum vertex cover in $G$ is called {\em vertex cover number}, denoted by $\vc(G)$.
The following proposition is a well-known relationship between treewidth, pathwidth, and vertex cover number.
\begin{proposition}\label{prop:graph_parameter}
For any graph $G$, it holds that $\tw(G)\le \pw(G)\le \vc(G)$.
\end{proposition}
 
In~\cite{Fomin2004}, Fomin and Thilikos proved that for any planar graph $G$, $\tw(G)\le 3.183\sqrt{n}-1$ and a tree decomposition of such width can be computed in polynomial time. Using this fact, we obtain the following proposition for $p$-apex graphs.
\begin{proposition}\label{prop:apex:tw}
Let $p$ be some constant. For any $p$-apex graph $G$, $\tw(G)\le 3.183\sqrt{n}+p-1$. Moreover, a tree decomposition of such width can be computed in polynomial time.
\end{proposition}
\begin{proof}
In ~\cite{Fomin2004}, Fomin and Thilikos proved that for any planar graph $G$, $\tw(G)\le 3.183\sqrt{n}-1$ and a tree decomposition of such width can be computed in polynomial time.
Thus, we first guess $p$ vertices such that $G$ becomes planar by deleting them.
Since we can check whether a graph is planar in time $O(n^2)$~\cite{Kawarabayashi2012},
this can be done in polynomial time by using the brute forth.
Now, we have a planar graph $G'$ obtained from $G$ by deleting such $p$ vertices.
Then we compute a tree decomposition of width $\tw(G')\le 3.183\sqrt{n}-1$ in polynomial time. 
Finally, we add $p$ vertices in $V(G)\setminus V(G')$ to each bag of a tree decomposition.
The width of such a tree decomposition of $G$ is clearly at most $3.183\sqrt{n}+p-1$.
\end{proof}
Proposition \ref{prop:apex:tw}  implies that there is a $2^{O(\sqrt{n}\log n)}$-time algorithm  for any $p$-apex graph if there is a ${\tw}^{O(\tw)}$-time or even an $n^{O(\tw)}$-time algorithm. Therefore, \textsc{Max Utilitarian} and \textsc{Max Egalitarian} with restricted weights can be solved in time  $2^{O(\sqrt{n}\log n)}$ on $p$-apex graphs from Theorems \ref{thm:treewidth:hedonic} and \ref{thm:pseudo_tw:Egal}.

\subsection{Problem list}
In this subsection, we list problems used for the proofs in this paper.
\begin{itemize}
\item \textsc{Max $k$-Cut}: Given an undirected and edge-weighted graph $G=(V,E,w)$, find a partition $(V_1, V_2, \dots ,V_k)$ that maximizes $\sum_{u_1 \in V_i, u_2 \in V_j, V_i \neq V_j} w_{u_1 u_2}$.
\textsc{Max 2-Cut} is known as \textsc{Max-Cut}. 

\item \textsc{$k$-Coloring}:  Given  an undirected graph $G=(V,E)$, determine whether there is a coloring $c: V\rightarrow \{1, \ldots, k\}$ such that $c(u)\neq c(v)$ for every $(u,v)\in E$. 

\item \textsc{Partition}:  Given a finite set of integers $A=\{a_1,a_2,\ldots, a_{n}\}$ and $W=\sum_{i=1}^{n} a_i$, determine whether there is partition $(A_1, A_2)$ of $A$ where $A_1\cup A_2=A$ and  $\sum_{a\in A_1}a = \sum_{a\in A_2}a=W/2$.

\item \textsc{$3$-Partition}:  Given a finite set of integers $A=\{a_1, \ldots, a_{3n}\}$, determine whether there is partition $(A_1, \ldots, A_n)$ such that $|A_i|=3$ and  $\sum_{a\in A_i}a=B$ for each $i$ where $B=\sum_{a\in A} a/n$. 

\end{itemize}
\section{Nash-stable}

Any symmetric ASHG always has a Nash-stable partition by Proposition \ref{prop:stability_max}.
However,  finding a Nash-stable solution is PLS-complete~\cite{gairing2010computing}.
In this section, we prove that \textsc{Nash-Stable}  is PLS-complete even on bounded degree graphs. 


\begin{figure}[tbp]
    \centering
    \includegraphics[width=8cm]{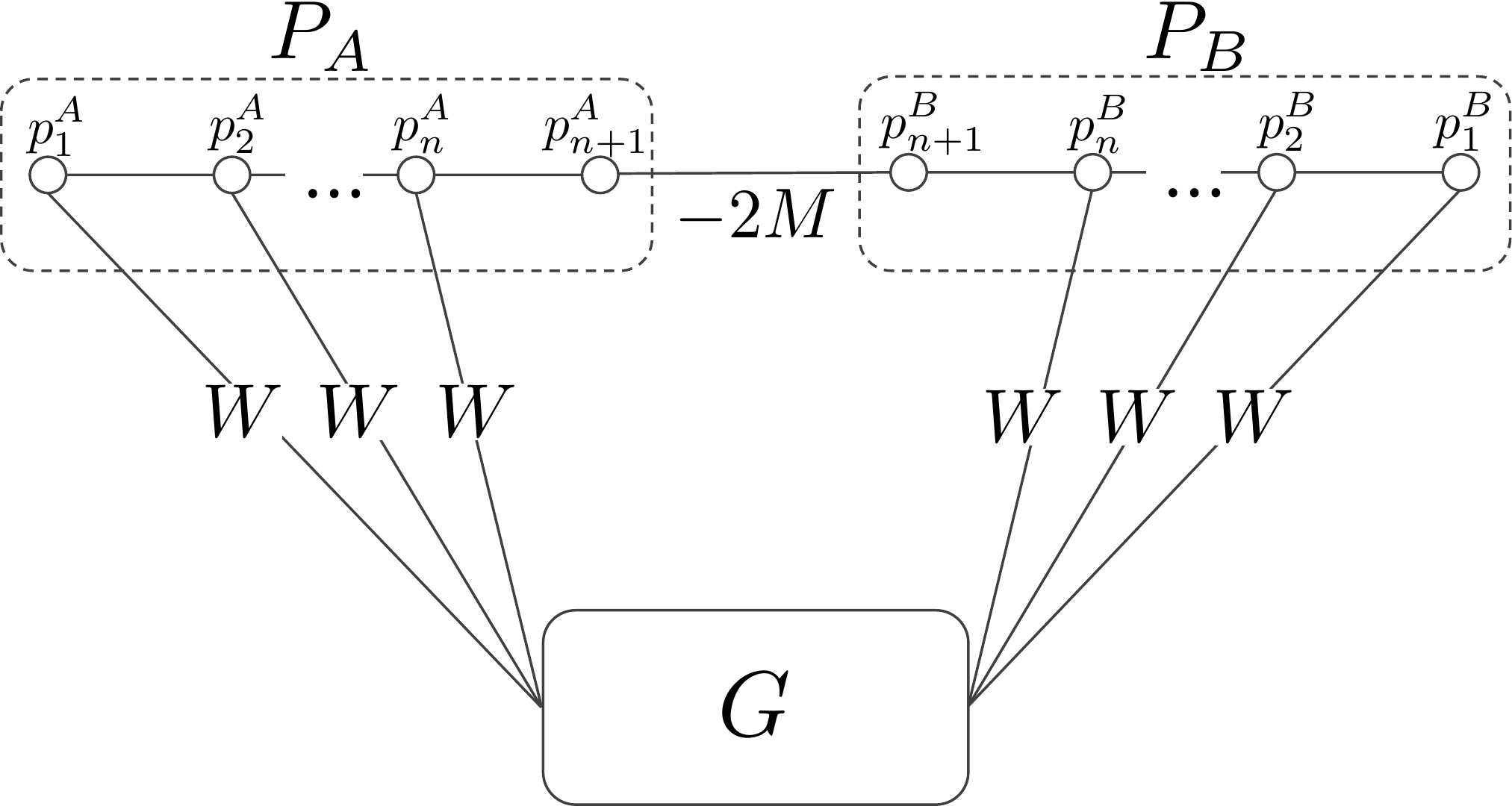}
    \caption{The constructed graph $G'$ in the proof of Theorem \ref{thm:md7}.}
    \label{fig:PLS}
\end{figure}
\begin{theorem}\label{thm:md7}
Symmetric \textsc{Nash-stable}  is PLS-complete even on graphs with maximum degree $\Delta= 7$.
\end{theorem}
\begin{proof}
We give a reduction from \textsc{Local Max-Cut} with flip.
\textsc{Local Max-Cut} is a local search problem of \textsc{Max-Cut}. In the flip neighborhood, two solutions are neighbors if one can be obtained from the other by moving one element to the other set.
\textsc{Local Max-Cut} with flip is PLS-complete on graphs with maximum degree $\Delta= 5$~\cite{elsasser2011settling}.

Given an edge weighted graph $G=(V,E,w)$, we construct $G'=(V',E',w')$ as follows.
Let $W=\sum_{e \in E} |w_e|+1$ and $M= nW +1$.  
First we set $w'_e=-w_e$ for every $e\in E$.
Then we add  $P_A=\{p^A_1, \ldots, p^A_{n+1}\}$ and $P_B=\{p^B_1, \ldots, p^B_{n+1}\}$ that form paths of length $n$, respectively. 
For $1\le i\le n$, we define $w'_{p^A_i p^A_{i+1}} =w'_{p^B_i p^B_{i+1}}=i (W +1 )$ as the weight of $\{p^A_i, p^A_{i+1}\}$ and $\{p^B_i, p^B_{i+1}\}$.
We connect each $p_i^A$ and $p_i^B$ to $v_i\in V$ by an edge of weight  $W$, respectively. 
Finally, we connect $p_{n+1}^A$ and  $p_{n+1}^B$ by an edge of weight $-3M$.
Note that $P_A \cup P_B$ forms a path of length $2n+1$.
We can observe that  the degree of a vertex in $V$ is at most $7$ and in $P_A \cup P_B$ is at most $3$.

In any Nash-stable partition in $ G^\prime $, $ p_{n + 1}^A $ and $ p_{n + 1}^B $ must be in different coalitions.
If not so,  the utility of $ p_ {n + 1} ^ A $ is at most $ -3 M + n (W + 1) ＋W <0 $, and it has an incentive to deviate to a singleton because the utility in a singleton is $0$.
Moreover, in any Nash-stable partition of $ G^\prime $, every vertex in $P_A$ always belong to the same coalition. 
Otherwise, there is an agent $p^A_i $ with utility at most $ (i-1) (W + 1) + W = i (W + 1) -1 $. 
Then the utility of $p^A_i $ can be increased to at least $ i (W + 1) $ by deviating to the coalition to which $ p^A_{i + 1} $ belongs.
Similarly, every vertex in $P_B$ always belong to the same coalition in  any Nash-stable partition of $ G^\prime $.
Therefore, in any Nash-stable partition, there exist a coalition $ C_A $ containing all vertices included in $ P_A $ and a coalition $ C_B $ containing all vertices included in $ P_B $.
Furthermore, if a vertex in $V$ does not belong to neither $C_A$ nor $C_B$, the utility is at most $0$.
Since the utility of the vertex in $C_A$ or $C_B$ is more than $0$, it must belong to either $C_A$ or $C_B$.
Thus, any Nash-stable partition $ \PP ^ *$ in $ G^\prime $ has exactly two coalitions $C_A$ containing $P_A$ and $C_B$ containing $P_B$.

Here, we can observe that any Nash-stable partition in $ G^\prime $ is a local optimal solution of \textsc {Local Max-Cut} in $G$. If not so, there is a vertex $ v \in V $ that can increase the weight of a cut in $G$ by flipping $v$ from the current set to the other.
In $G^\prime $, such a vertex deviates to  the other coalition because the utility increases.
Conversely,  we are given a local optimal cut $(C_1, C_2)$ of \textsc {Local Max-Cut}. Then a partition $(C_1\cup P_A, C_2\cup P_B)$ in $G'$ is a \textsc{Nash-stable} partition.
This is because any vertex in $ P_A$ and  $P_B $  does not deviate from the current coalition.
Moreover, suppose that there is a vertex $v$ in $ C_1 \cup C_2 =V $  deviates to the other coalition to increase the utility. This implies that there is a vertex $v$ in $G$ that increases the weight by flipping $v$.
This contradicts the optimality of $(C_1, C_2)$.
\end{proof}

\section{Max Utilitarian}

\begin{theorem}\label{thm:mu_3apex}
\textsc{Max Utilitarian} is strongly NP-hard on 3-apex graphs even if the preferences are symmetric.
\end{theorem}
\begin{proof}
Let $G=(V,E)$ be an unweighted graph where $|V|=n$ and $|E|=m$.
We first confirm that \textsc{Max $3$-Cut} is NP-hard for planar graphs.
This can be done by a reduction from \textsc{$3$-coloring} on planar graphs which is NP-complete~\cite{GJS76}.
If an unweighted graph $G$ is 3-colorable, it is clear that $G$ has a 3-cut of size $m$ because for every edge, its endpoints have different colors. 
On the other hand, if an unweighted graph $G$ has a 3-cut of size $m$, it is obviously 3-colorable.

Then we give a reduction from \textsc{Max $3$-Cut} to  \textsc{Max Utilitarian}.
Given an unweighted planar graph $G=(V,E)$ of an instance of  \textsc{Max $3$-Cut}, we add three super vertices $S=\{s_1,s_2,s_3\}$ such that and each $s_i$ is connected to all vertices in $G$ by three edges of weight $m$.
Moreover, we connect $s_1,s_2,s_3$ to each other by edges of weight $-6mn$, and hence $S$ forms a clique.
Finally, for each edge $e\in E$, we define the weight  $w_e=-1$.
Let $G'$ be the constructed graph.

In the following, we show that there is a 3-cut of size $k$ in $G$ if and only if there is a partition $\PP$ with social welfare $2(3mn-m+k)$.
Given a 3-cut $(V_1,V_2,V_3)$ of size $k$ in $G$, we construct a partition $\PP=\{V_1\cup \{s_1\},V_2\cup \{s_2\},V_3\cup \{s_3\}\}$.
Since the number of edges in $E$ in coalitions is $m-k$, any edge in  $G[S]$ is not contained in coalitions, and every edge between $s\in S$ and $v\in V$ is contained in coalitions, the social welfare of $\PP$ is $2(3mn-m+k)$.

Conversely, we are given a partition $\PP$ of social welfare $2(3mn-m+k)$.
If a coalition contains at least two vertices in $S$, the social welfare is at most $6m-6mn<0$.
Thus, each $s_i$ does not belong to the same coalition.
If there is $v\in V$ that does not belong to a coalition containing $s\in S$, then the social welfare is at most $2\cdot 3m(n-1)=6mn-6m<2(3mn-m+k)$.
Hence, $V$ must be partitioned into three sets adjacent to either $s_1$, $s_2$ or $s_3$.
Since the social welfare of $\PP$ is $2(3mn-m+k)$, every edge between $s\in S$ and $v\in V$ is contained in a coalition, and the weight of an edge in $E$ is $-1$, the number of edges in $E$ in coalitions is $m-k$.
This implies that there is a 3-cut of size $k$. 
\end{proof}

\begin{theorem}\label{thm:treewidth:hedonic}
Given a tree decomposition of width ${\tw}$, \textsc{Max Utilitarian} can be solved in time ${\tw}^{O(\tw)} n$.
\end{theorem}
\begin{proof}
Our algorithm is based on  dynamic programming on a tree decomposition for connectivity problems such as \textsc{Steiner tree} \cite{Cygan2015}. 
In our dynamic programming, we keep track of all the partitions in each bag.

We define the recursive formulas for computing the social welfare of each partition ${\mathcal P}$ in the subgraph based on a subtree of a tree decomposition. 
Let   ${\mathcal P_i}$ be a partition of $X_i$.
We denote by $A_i[{\mathcal P}_i]$ the maximum social welfare in the subgraph $G_i$ such that $X_i$ is partitioned into ${\mathcal P_i}$. 
Notice that $A_r[\emptyset]$ in root node $r$ is the maximum social welfare of $G$.
We denote a parent node by $i$ and its child node by $j$. For a join node, we write $j_1$ and $j_2$ to denote its two children. 

\paragraph*{Leaf node:}  
In leaf nodes, we set $A_i[\emptyset]=0$.

\paragraph*{Introduce vertex $v$ node:} 
In an introduce vertex $v$ node $i$, let $C_v\in {\mathcal P}_i$ be the coalition including $v$. 
We notice that the social welfare is increased by edges between $v$ and vertices in the coalition including $v$. Thus, the recursive formula is defined as: 
$A_i[{\mathcal P}_i]=A_j[{\mathcal P}_j]+\sum_{u\in N(v)\cap C_v} w_{uv}+\sum_{u\in N(v)\cap C_v} w_{vu}$
where ${\mathcal P}_j = {\mathcal P}_i \setminus \{C_v\} \cup \{C_v\setminus \{v\}\}$.

\paragraph*{Forget $v$ node:}
In a forget $v$ node, we only take a partition with maximum social welfare when we forget $v$ because $v$ does not affect the social welfare hereafter.
Thus, the recursive formula is defined as: 
$A_i[{\mathcal P}_i]=\max_{{\mathcal P}_j\in {\mathcal D}_j} A_j[{\mathcal P}_j]$ 
where ${\mathcal D}_j=\{{\mathcal P}_j \mid {\mathcal P}_j \setminus \{C_v\} \cup \{C_v\setminus \{v\}\}={\mathcal P}_i\}$.

\paragraph*{Join node:}
A join node $i$ has two child nodes $j_1, j_2$ where $X_i=X_{j_1}=X_{j_2}$.
The social welfare of each partition $X_i$ is the sum of the corresponding partition of $X_{j_1}=X_{j_2}$.
Therefore,  the recursive formula for a join node is defined as:
$A_i[{\mathcal P}_i]=A_{j_1}[{\mathcal P}_{i}]+A_{j_2}[{\mathcal P}_{i}]-\sum_{C\in \PP_i} \sum_{u,v\in C} (w_{uv}+w_{vu})$.
The last term means subtracting the double counting of edges.

Because the size of each DP table is ${\tw}^{O(\tw)}$, we can compute the recursive formulas in time ${\tw}^{O(\tw)}$.
As the result,  the total running time is ${\tw}^{O(\tw)}n$. 
\end{proof}

By Proposition \ref{prop:stability_max},  symmetric \textsc{Nash-stable} is also solvable in time  ${\tw}^{O(\tw)} n$.
\begin{corollary}\label{cor:Nash-degree2}
Given a tree decomposition of width ${\tw}$, symmetric \textsc{Nash-stable} can be solved in time ${\tw}^{O(\tw)} n$.
\end{corollary}

\section{Max Envy-free and Max Egalitarian}

In~\cite{Aziz2013}, Aziz et al. show that finding an envy-free partition is trivial because a partition of singletons is envy-free.
However, finding a {\em maximum} envy-free partition is much more difficult than  finding an envy-free partition.

\begin{theorem}\label{thm:vc:Envy-free}
\textsc{Max Envy-free} is weakly NP-hard on series-parallel graphs of vertex cover number $2$ even if the preferences are symmetric.
\end{theorem}

\begin{figure}[tbp]
    \centering
    \includegraphics[width=6.5cm]{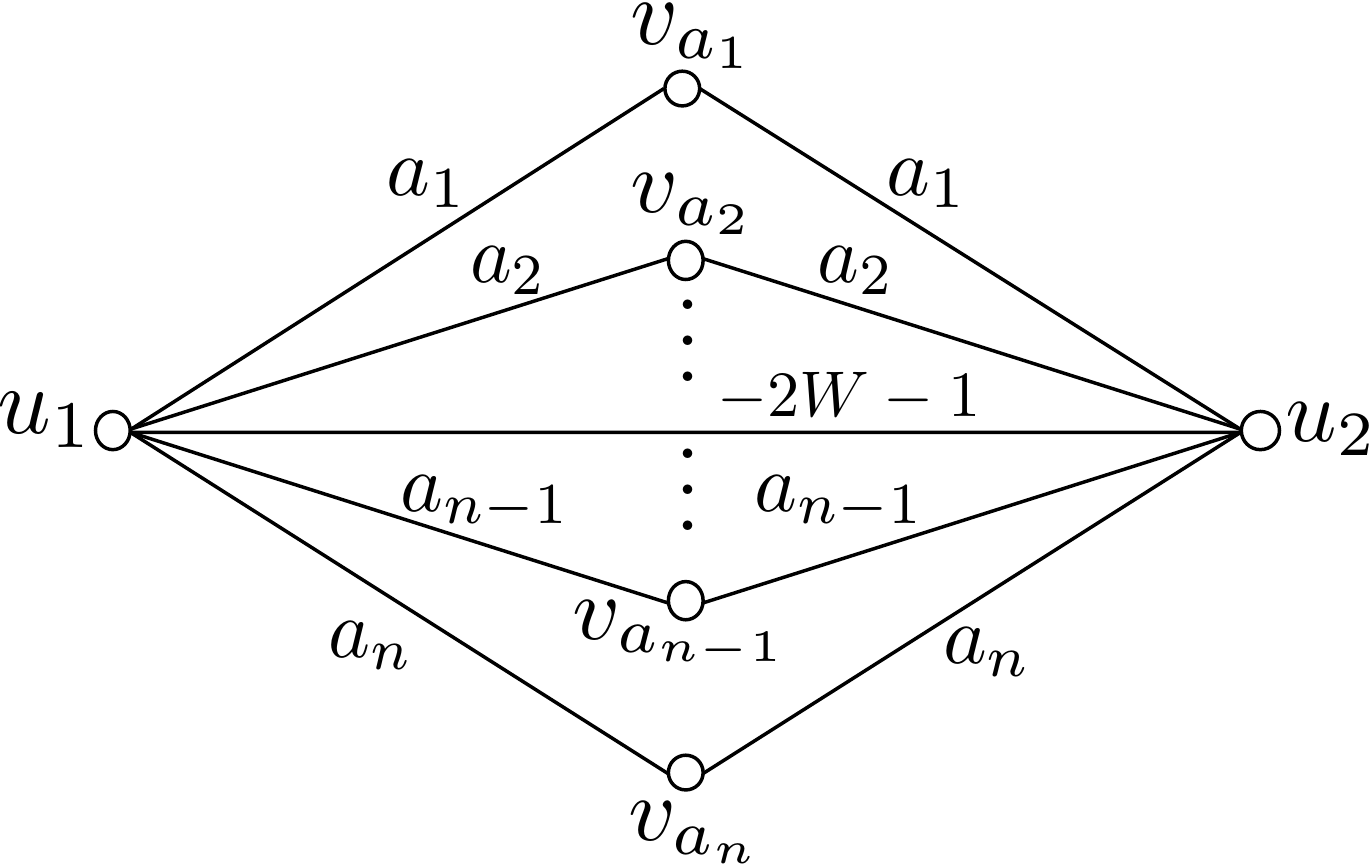}
    \caption{The constructed graph $H$.}
    \label{fig:Envy-free_Partition}
\end{figure}

\begin{proof}
We give a reduction from \textsc{Partition}, which is weakly NP-complete \cite{GJ1979}.
Without loss of generality, we suppose $a_1\le a_2\le\ldots\le a_n$ and $a_n< W/2$. 

Given a set of positive integers $A=\{a_1,a_2,\ldots,a_{n}\}$, we build the corresponding vertex set $V_A=\{v_{a_1}, v_{a_2}, \ldots, v_{a_n}\}$. Then we construct an edge-weighted complete bipartite graph $K_{2,n}=(V_A\cup U,E)$ where  $U=\{u_1,u_2\}$.
For each edge $\{v_{a_i},u_j\}\in E$, we set the weight $w_{v_{a_i}u_j}=a_i$.
Finally, we add  $e_u=\{u_1,u_2\}$ of weight $-2W-1$.
Let $H=(V_A\cup U, E\cup\{e_u\})$ be the constructed graph.
Note that $H$ is a series parallel graph and $\vc(H)=2$ (see Fig. \ref{fig:Envy-free_Partition}).

We show that an instance of \textsc{Partition} is a yes-instance if and only if there is an envy-free partition with social welfare at least $2W$ in $H$. 

Given a partition $(A_1, A_2)$ of $A$ such that $\sum_{a\in A_1}a = \sum_{a\in A_2}a = W/2$, let $V_{A_1}$ and $V_{A_2}$ be the corresponding vertex set to $A_1$ and $A_2$, respectively. In short, $V_{A_{i}}=\{v_{a}\mid a\in A_{i}\}$ for $i\in \{1,2\}$. 
Let  $\PP=\{V_{A_1}\cup \{u_1\}, V_{A_2}\cup \{u_2\}\}$ be a partition in $H$.
By the definition of $H$, we have $\ut_{\PP}(v_a)=a$ for every $v_a\in V_A$ and $\ut_{\PP}(u)=W/2$ for every $u\in U$.
Let $C_i=V_{A_i}\cup \{u_i\}$ for $i\in \{1,2\}$.
For an agent $v_a\in C_i$, consider  $C'_i=C_j\setminus \{w\} \cup \{v_a\}$ for any $w\in C_j\neq C_i$.
In this case, the utility of $v_a$ is at most $a$.
Moreover, consider $C'_1=C_2\setminus \{w\} \cup \{u_1\}$ for any $w\in C_2$.
Then the utility of $u_1$ is at most $W/2$.
Therefore, $\PP$ is envy-free.
Also, the social welfare of $\PP$ is $W/2+W/2+\sum_{a\in A} a= 2W$.

Conversely, we are given an envy-free partition $\PP$ with social welfare at least $2W$ in $H$. 
If $\PP$ has a coalition that contains both $u_1$ and $u_2$, the social welfare of $\PP$ is strictly less than $2W$.
Thus, there are two coalitions $C_1$, $C_2$ in $\PP$ such that $u_1\in C_1$ and $u_2\in C_2$.
For each $a\in A$,  one of $v_a\in C_1$, $v_a\in C_2$, and $v_a\in C'$ holds where $C'\subseteq V_A$ does not contain $u_1$ and $u_2$.
This implies that the social welfare of $\PP$ is at most $2W$ because at most one of $\{v_a, u_1\}$ and $\{v_a, u_2\}$ contributes to the social welfare.
If there is an agent $v_a$  in neither $C'$, the social welfare of $\PP$ is strictly less than $2W$. 
Therefore, it holds that either $v_a\in C_1$ or $v_a\in C_2$ for each $v_a\in V_A$ in  $\PP$.
Suppose that $\ut_{\PP}(u_1)>\ut_{\PP}(u_2)$. Then $u_2$ envies $u_1$ by the definition of $H$.
This contradicts the assumption.
Now, we have $\ut_{\PP}(u_1)=\ut_{\PP}(u_2)=W/2$.
Let ${A_i}=\{a\in A\mid v_a\in N(u_i)\subseteq V_A\}$ for $i\in \{1,2\}$.
Then a partition $\{A_1, A_2\}$ satisfies that $\sum_{a\in A_1}a = \sum_{a\in A_2}a=W/2$.
\end{proof}


\begin{figure}[tbp]
 \begin{minipage}{0.55\hsize}
    \centering
    \includegraphics[width=8.0cm]{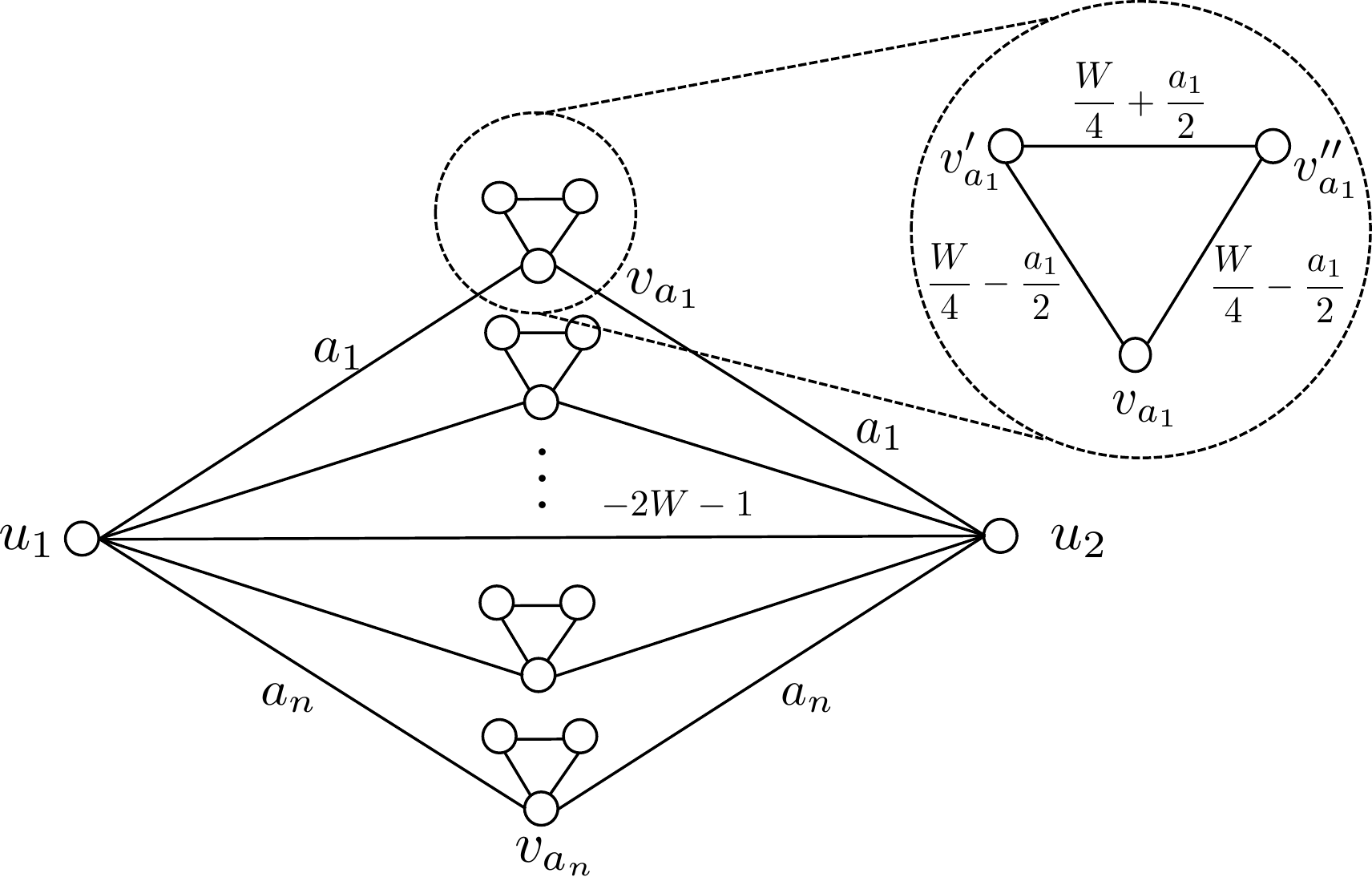}
    \caption{The constructed graph $H'$.}
    \label{fig:Egal_Partition_tw}
 \end{minipage}
 \begin{minipage}{0.47\hsize}
    \centering
    \includegraphics[width=6.1cm]{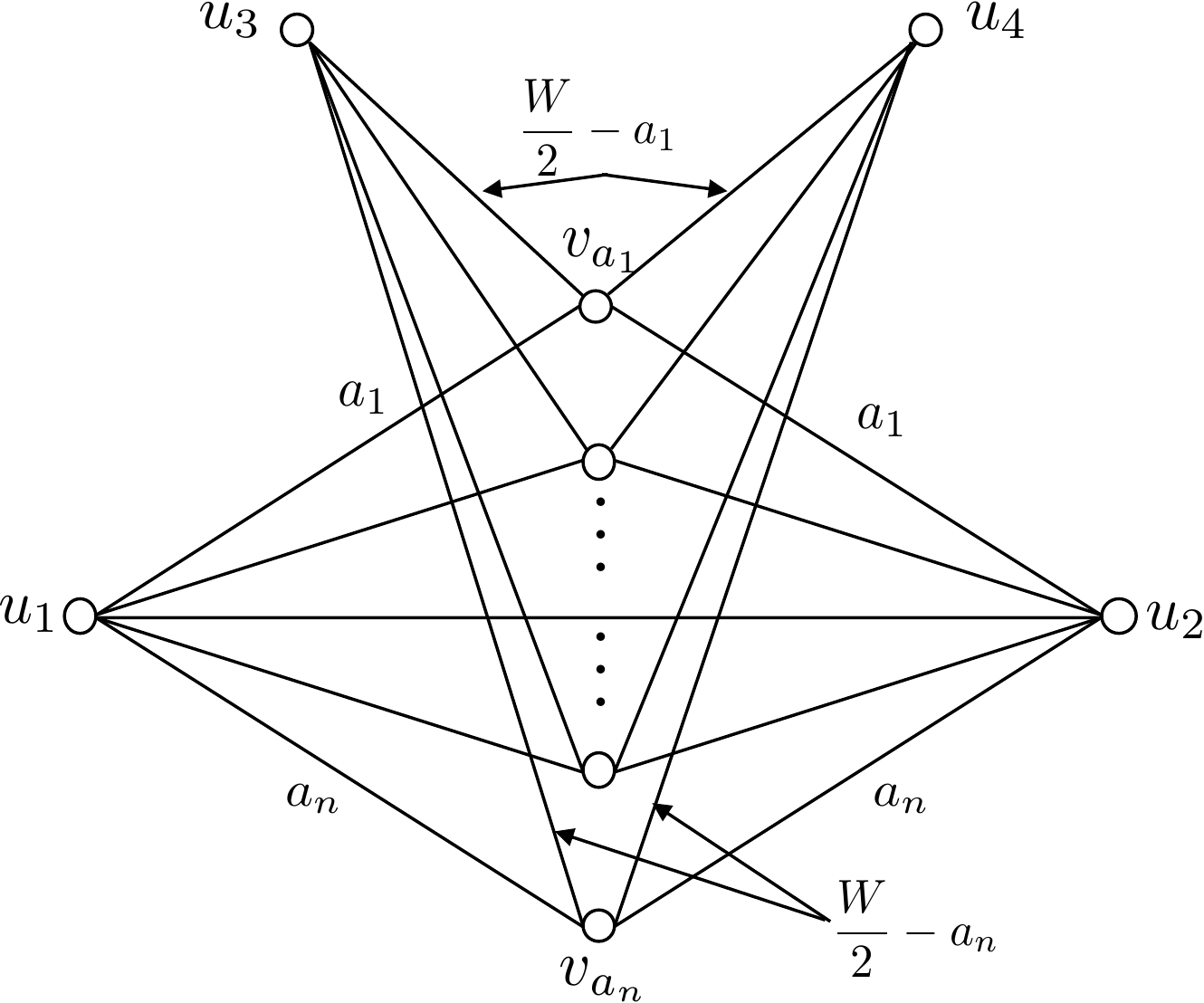}
    \caption{The constructed graph $H''$.}
    \label{fig:Egal_Partition_vc}
 \end{minipage}
\end{figure}


Next, we show that \textsc{Max Egalitarian}  is weakly NP-hard on series-parallel graphs of pathwidth $4$. 
Note that the  class  of  series-parallel graph is equivalent to graphs with treewidth $2$.
\begin{theorem}\label{thm:tw:Egal}
In the symmetric hedonic games, \textsc{Max Egalitarian} is weakly NP-hard on series-parallel graphs of pathwidth $4$  even if the preferences are symmetric.
\end{theorem}
\begin{proof}
We give a reduction from \textsc{Parition} as in the proof of Theorem~\ref{thm:vc:Envy-free}, 
though we adopt a bit different graph from $H$. 
For each $v_{a_i}$ in $V_A$, we create two copies of $v_{a_i}$, denoted by $v'_{a_i}$ and $v''_{a_i}$ respectively, such that they form a clique.
For each $a_i\in A$, we define the weights $w_{v_{a_i} v'_{a_i}}=w_{v_{a_i} v''_{a_i}}=W/4-a_i/2$ and the weight  $w_{v'_{a_i} v''_{a_i}}=W/4+a_i/2$.
Without loss of generality, we can assume that each $a_i$ is even.
Let $V'_A=\{v'_a\mid a\in A\}$ and $V''_A=\{v''_a\mid a\in A\}$ and let $H'$ be the constructed graph (see Fig. \ref{fig:Egal_Partition_tw}).
For the graph $H'$, if we set $X_i=\{u_1, u_2, v_{a_i}, v'_{a_i},v''_{a_i}\}$ for each $a_i$ and connect $X_i$ and $X_{i+1}$ by an edge, we can construct a path decomposition of width at most $4$.
Also, it is easy to show that $H'$ is a series-parallel graph, and hence the treewidth of $H'$ is $2$.

 In the following, we show that there exists a partition $(A_1, A_2)$ such that $\sum_{a\in A_1}a = \sum_{a\in A_2}a = W/2$ if and only if there exists a partition $\PP'$ in $H'$ such that $\min_{v\in V(H')}\ut_{\PP'}(v)= W/2$.

Given a partition $(A_1, A_2)$ of $A$ such that $\sum_{a\in A_1}a = \sum_{a\in A_2}a = W/2$, we set $V_{A_i}=\{v_{a}\in V_A\mid a\in A_i\}$,  $V'_{A_i}=\{v'_{a}\in V'_A\mid a\in A_i\}$, and  $V''_{A_i}=\{v''_{a}\in V''_A\mid a\in A_i\}$ for $i\in \{1, 2\}$.
Let $\PP=\{V_{A_1}\cup V'_{A_1}\cup V''_{A_1}\cup \{u_1\}, V_{A_2}\cup V'_{A_2}\cup V''_{A_2} \cup\{u_2\}\}$ be a partition in $H'$.
By the definition of $H'$, we have $\ut_{\PP}(v)=W/2$ for every $v\in V(H')$.

Conversely, we are given a partition $\PP$ such that $\min_{v\in V(H')}\ut_{\PP'}(v)= W/2$.
If $\PP$ has a coalition that contains both $u_1$ and $u_2$, the utilities of $u_1$ and $u_2$ are less than $0$.
Thus, we suppose that there are two coalitions $C_1$, $C_2$ in $\PP$ such that $u_1\in C_1$ and $u_2\in C_2$.
For each $a\in A$, if $v_a$, $v'_a$, and $v''_a$ belong to a different coalition from the other two, the utility of each is strictly less than $W/2$.
Moreover, if $v_a\in V_A$ belongs to neither $C_1$ nor $C_2$, $\ut_{\PP}(v_a)<W/2$.
Thus, $v_a$, $v'_a$, and $v''_a$ belong to the same coalition of either $C_1$ or $C_2$.

Since the utilities of $u_1$ and $u_2$ are $W/2$, it holds  that $\sum_{v\in N(u_1)} w_{u_1v}=\sum_{v\in N(u_2)} w_{u_2v}=W/2$.
Therefore, if we set $A_i=\{a\in A\mid v_a\in N(u_i)\}$, $(A_1, A_2)$ is a partition satisfying   $\sum_{a\in A_1}a = \sum_{a\in A_2}a = W/2$.
\end{proof}

Note that the pathwidth and the treewidth of $H'$ are bounded, but the vertex cover number is not bounded.
We can similarly show 
that \textsc{Max Egalitarian} is also weakly NP-hard on bounded vertex cover number graphs by using the reduced graph $H''$ in Fig.~\ref{fig:Egal_Partition_vc}.

\begin{theorem}\label{thm:vc:Egal}
\textsc{Max Egalitarian} is weakly NP-hard on 2-apex graphs of vertex cover number $4$ even if the preferences are symmetric.
\end{theorem}
\begin{proof}
We give a reduction from \textsc{Parition} as in the proof of Theorem~\ref{thm:vc:Envy-free}.
Without loss of generality, we suppose  $a_{n}<W/2$.
To show this, we modify the graph $H'$ in Theorem~\ref{thm:vc:Envy-free}.
We add two super vertices $u_3$ and $u_4$ connecting to all vertices in $V_A$, respectively. 
For each edge $\{a_i, u_j\}$ where $i\in \{1,\ldots, n\}$ and $j\in \{3,4\}$, we set the weight $w_{a_i u_j}=W/2-a_i$.
Let $H''$ be the constructed graph (see Fig. \ref{fig:Egal_Partition_vc}).
We can observe that $H''$ is a 2-apex graph because a graph obtained from $H''$ by deleting $u_3,u_4$ is planar.
Moreover, $\{u_1,u_2,u_3,u_4\}$ is a vertex cover of size four in $H''$.

Then, we show that there exists a partition $(A_1, A_2)$ such that $\sum_{a\in A_1}a = \sum_{a\in A_2}a = W/2$ if and only if there exists a partition $\PP$ in $H''$ such that $\min_{v\in V(H')} \ut_{\PP}(v)= W/2$.

Given a partition $(A_1, A_2)$ of $A$ such that $\sum_{a\in A_1}a = \sum_{a\in A_2}a = W/2$, we set $V_{A_i}=\{v_{a}\in V_A\mid a\in A_i\}$ for $i\in \{1, 2\}$.
Let $\PP=\{V_{A_1}\cup \{u_1,u_3\}, V_{A_2}\cup\{u_2,u_4\}\}$ be a partition in $H''$.
By the definition of $H'$, we have $\ut_{\PP}(v)=W/2$ for every $v\in V(H')$.
For $j\in \{1,2\}$, it holds that $\ut_{\PP}(u_j)=\sum_{a\in A_j} a=W/2$.
Since   $a_{n}<W/2$, it holds that  $|A_1|, |A_2|\ge 2$  and then we have $\ut_{\PP}(u_j)\ge W/2 $.
Moreover, since any $v_a$  has either $u_1$ and $u_3$ or $u_2$ and $u_4$ as neighbors, $\ut_{\PP}(v_a)=W/2 $ holds.

Conversely, we are given a partition $\PP$ such that $\min_{v\in V(H'')}\ut_{\PP}(v)= W/2$.
If $\PP$ has a coalition that contains both $u_1$ and $u_2$, the utilities of $u_1$ and $u_2$ are less than $0$.
We suppose that there are two coalitions $C_1$, $C_2$ in $\PP$ such that $u_1\in C_1$ and $u_2\in C_2$.

Since the utilities of $u_1$ and $u_2$ are $W/2$, it holds  that $\sum_{v\in N(u_1)} w_{u_1v}=\sum_{v\in N(u_2)} w_{u_2v}=W/2$.
Therefore, if we set $A_i=\{a\in A\mid v_a\in N(u_i)\}$, $(A_1, A_2)$ is a partition satisfying   $\sum_{a\in A_1}a = \sum_{a\in A_2}a = W/2$.
Note that $N(u_1)=N(u_2)\subseteq V_A$.
\end{proof}

Aziz et al. show that {\em asymmetric} \textsc{Max Egalitarian} is strongly NP-hard~\cite{Aziz2013}.
We show that  {\em symmetric} \textsc{Max Envy-free} and {\em symmetric} \textsc{Max Egalitarian} remain to be strongly NP-hard.
To show this, we give a reduction from \textsc{3-Partition}, which is strongly NP-complete~\cite{GJ1979}.
\begin{theorem}\label{thm:Envy-Egal:strongNP}
\textsc{Max Envy-free} and \textsc{Max Egalitarian}  are strongly NP-hard even if the preferences are symmetric.
\end{theorem}
\begin{proof}
We first explain a reduction for \textsc{Max Envy-free}.
For $H$ in Theorem \ref{thm:vc:Envy-free}, we set $n=3m$. Then we add $m-2$ super vertices $\{u_3, \ldots, u_{m}\}$ connecting to every vertex $v_{a_i}$ in $V_A$ by an edge of weight $a_i$.
Moreover, we connect $u_p$ and $u_q$ with an edge of weight $-mB$ for $p, q\in \{1, \ldots, m\}$ and $p\neq q$ where $B=\sum_{a\in A}a/n$. Note that $\{u_1, \ldots, u_{m}\}$ forms a clique. 
The number of vertices in the constructed graph is $m+3m=4m$.
It is easily seen that there is a partition  $(A_1, \ldots, A_{m})$ such that $|A_i|=B$ for each $i$ if and only if there is an envy-free partition with social welfare at least $2mB$ in $H$ as in the proof of Theorem \ref{thm:vc:Envy-free} . 

Next, we explain a reduction for \textsc{Max Egalitarian}.
For  $H'$ in Theorem \ref{thm:tw:Egal}, we set $n=3m$.
We create $m$ super vertices $\{u_3, \ldots, u_{m}\}$ connecting to every vertex $v_{a_i}$ in $V_A$ by an edge of weight $a_i$ in $H'$.
Moreover, we connect $u_p$ and $u_q$ with an edge of weight $-mB$ for  $p, q\in \{1, \ldots, m\}$ and $p\neq q$.
Finally, we change the weights of $\{v_{a_i}, v'_{a_i}\}$ and $\{v_{a_i}, v'_{a_i}\}$ to $B/2-a_i/2$ and the weight of $\{v'_{a_i}, v''_{a_i}\}$ to $B/2+a_i/2$.
Then we can observe that there is a partition  $(A_1, \ldots, A_{m})$ such that $|A_i|=B$ for each $i$ if and only if there is a partition partition in $H'$ such that $\min_{v\in V(H')}\ut_{\PP'}(v)= B$  as in the proof of Theorem \ref{thm:tw:Egal}. 
\end{proof}

Since $\tw(G)\le \vc(G)$, \textsc{Max Envy-free} is weakly NP-hard on graphs of $\tw(G)=2$ by Theorem \ref{thm:vc:Envy-free}. 
Also, \textsc{Max Egalitarian} is weakly NP-hard on graphs of $\tw(G)= 2$ by  Theorem \ref{thm:tw:Egal}.
However, we show that symmetric \textsc{Max Envy-free} and symmetric \textsc{Max Egalitarian}  on trees, which are of treewidth $1$,  are solvable in linear time.
Indeed, we can find an envy-free and maximum egalitarian partition with maximum social welfare.
Such a partition consists of connected components of a forest obtained by removing all negative edges from an input  tree.
\begin{theorem}\label{thm:tree:Egal}
Symmetric \textsc{Max Envy-free} and symmetric \textsc{Max Egalitarian} are solvable in linear time on trees.
\end{theorem}

Note that linear-time solvability does not hold for asymmetric cases, though asymmetric Max Egalitarian on trees can be solved in near-linear time.
\begin{theorem}\label{thm:asym:tree}
\textsc{Max Egalitarian} can be solved in time $O(n \log W)$ on trees. 
\end{theorem}
\begin{proof}
In this proof, we design an algorithm for \textsc{Max Egalitarian} on trees with self-loops, which is a slightly wider class of graphs. Given a tree $T$ with self-loops and a non-negative value $W$, our algorithm determines whether there exists a partition of $T$ such that the utility of every agent is at least $W$ in linear time. The idea is that  we can immediately answer ``No'' by focusing on a leaf, or we can reduce the given $T$ to a smaller tree $T'$ with self-loops whose answer is equivalent to the one for $T$. 

Given a tree $T$, we consider a leaf $u$ and its adjacent vertex $v$. Let $w(v)$,  $w(u,v)$ and $w(v,u)$ be the weights of self-loop of $v$, edges $(u,v)$ and $(v,u)$, respectively. We consider the following four cases: (i) $w(u)<W$ and $w(u)+w(u,v)< W$ hold, (ii) $w(u)<W$ and $w(u)+w(u,v) \ge W$ hold, (iii) $w(u)\ge W$ and $w(u)+w(u,v)<W$ hold, and (iv) $w(u)\ge W$ and $w(u)+w(u,v) \ge W$ hold. In case (i), $u$ can be isolated or can be with $v$, but in any cases, $u$'s utility is smaller than $W$; the answer is obviously ``No''. In case (ii), in order to give $u$ utility at least $W$, $u$ and $v$ should belong to a same coalition, which implies that $v$ receives utility $w(v,u)$. This can be interpreted that $u$ is contracted into $v$ and the weight of self-loop $(v,v)$ is updated to $w(v)+w(v,u)$. In case (iii), in order to guarantee at least $W$ utility for $u$, $u$ should be isolated. Then, we simply consider the problem for $T'$ obtained from $T$ by deleting $u$. In case (iv), $u$ can have utility at least $W$ whichever $u$ belongs to a same coalition with $v$. We then consider two subcases: (iv-1) $w(v,u)>0$ and (iv-2) $w(v,u)<0$ (if $w(v,u)=0$, $u$ does not affect the partition. We can ignore $v$). 
If (iv-1), $v$ can reserve utility $w(v,u)$ by belonging to a same coalition with $u$; we can apply the same argument with (ii).   
If (iv-2), it is better that $v$ does not belong to a same coalition with $u$; we can apply the same argument with (iii). 

By the above observation, we can immediately say ``No'', or obtain $T'$ with one smaller vertices. Since the above check can be done in $O(1)$, the decision problem can be done in $O(n)$ time, where $n$ is the number of vertices. By applying the binary search, we can obtain a maximum egalitarian partition. 
\end{proof}

Theorems \ref{thm:tw:Egal} and \ref{thm:vc:Egal} mean that \textsc{Max Egalitarian} is weakly NP-hard even on bounded treewidth graphs. On the other hand, we show that there is a pseudo FPT algorithm for \textsc{Max Egalitarian} when parameterized by treewidth.
\begin{theorem}\label{thm:pseudo_tw:Egal}
Given a tree decomposition of width $\tw$, \textsc{Max Egalitarian} can be solvable in time $(\tw W)^{O(\tw)}n$ where $W=\max_{u \in V} \sum_{v\in N(u)}|w_{uv}|$.
\end{theorem}
\begin{proof}
Let $V_i$ be the set of vertices in $X_i$ or the descent of $X_i$ on a tree decomposition.
Then we define DP tables of our dynamic programming.

Let $\PP_i$ be a partition of $X_i$ and ${\bf u}_i$ be a $|X_i|$-dimensional vector whose elements take from $-W$ to $W$, called a {\em utility vector} of $X_i$. 
For $v\in X_i$, the element ${\bf u}_i(v)$ represents the utility of $v$ in $G[V_i]$.
Finally, we define $A_i[\PP_i, {\bf u}_i]$ for each bag $X_i$ by using $\PP_i$ and ${\bf u}_i$ as the maximum minimum utility of an agent in $V_i\setminus X_i$ in $G[V_i]$.
The value of  $A_r[\emptyset, \emptyset]$ is an optimal value for \textsc{Max Egalitarian} in $G$.
In the following, we define the recursive formulas for computing $A_i[\PP_i, {\bf u}_i]$ on a nice tree decomposition.

\paragraph*{Leaf node: }  
We initialize DP tables for each leaf node $i$ as $A_i[\emptyset, \emptyset]=W+1$.
Note that the maximum minimum utility is at most $W$ and once we execute the recursive formula in a forget node,  $A_i[\PP_i, {\bf u}_i]$ becomes at most $W$.

\paragraph*{Introduce vertex $v$ node: }
Let $C_v\in \PP_i$ be  a coalition that contains $v$ in an introduce $v$ node $i$. 
Note that $C_v$ may contain only $v$, that is, $C_v=\{v\}$.
In an introduce $v$ node, an agent $v$ is added to a coalition.
This changes the utilities of agents in  $C_v$.
Also, the utility of $v$ in $G[V_i]$ is the sum of weight of edges between $v$ and agents in $C_v$.
Since every agent in $X_j$ also appears in $X_i$, the maximum minimum utility of an agent in $V_i\setminus X_i$ in $G[V_i]$ does not change.
Therefore, we define the recursive formula as follows:
$A_i[{\PP}_i, {\bf u}_i]=A_j[{\PP}_j, {\bf u}_j]$, 
where $\PP_j=\PP_i\setminus \{C_v\} \cup \{C_v\setminus \{v\}\}$, ${\bf u}_j(u)={\bf u}_i(u)-w_{uv}$ for $u\in C_v\setminus \{v\}$, ${\bf u}_j(u)={\bf u}_i(u)$ for other $u$'s in $X_i\setminus \{v\}$, and $\sum_{u\in N(v)\cap C_v}w_{vu}={\bf u}_i(v)$. Otherwise, we define $A_i[{\PP}_i, {\bf u}_i]=-\infty$ as an invalid case.

\paragraph*{Forget $v$ node: }
In a forget $v$ node, if a vertex $v$ is forgotten, it never appears in $X_i$ and its ancestors on the decomposition tree.
This implies that the utility of $v$ does not change hereafter.
Namely,  the maximum minimum utility among forgotten agent  is stored in $A_i[{\PP}_i, {\bf u}_i]$ in some sense. Thus what we need to do here is to update the minimum by comparing the previous maximum minimum utility with the utility of the newly forgotten agent, which can be the new minimum. Taking the maximum among $\PP_j$ and ${\bf u}_j$, this can be interpreted as the following recursive formula:
\[
A_i[\PP_i, {\bf u}_i]=\max_{\PP_j, {\bf u}_j}\min \{A_j[\PP_j, {\bf u}_j], {\bf u}_j(v)\},
\]
where ${\bf u}_j(u)={\bf u}_i(u)$ for $u\in X_i$  and $\PP_j\setminus \{C_v\}\cup \{C_v\setminus \{v\}\}=\PP_i$. 
The condition  $\PP_j\setminus \{C_v\}\cup \{C_v\setminus \{v\}\}=\PP_i$ means that the coalition to which an agent  belongs in node $j$ is the same as the coalition to which an agent  belongs in node $i$.

\paragraph*{Join node: }
For two children $j_1, j_2$ of a join node $i$, it holds that $X_i=X_{j_1}=X_{j_2}$.
To update $A_i[{\PP}_i, {\bf u}_i]$  in a join node, we first take the minimum of $A_{j_1}[{\PP}_i, {\bf u}_{j_1}]$  and $A_{j_2}[{\PP}_i, {\bf u}_{j_2}]$. 
Note that the maximum minimum utility among forgotten agent until $X_i$ is the minimum of ones until the children nodes. 
Here, for every agent $v\in X_i$, ${\bf u}_i(v)={\bf u}_{j_1}(v)+{\bf u}_{j_2}(v)-\sum_{u\in N(v)\cap C_v}w_{vu}$ must hold. The subtraction avoids the double counting of edges.
Then taking the maximum among ${\bf u}_{j_1}$ and ${\bf u}_{j_2}$ satisfying the above condition,
the recursive formula can be defined as follows:
$$A_i[{\PP}_i, {\bf u}_i]=\max_{{\bf u}_{j_1}+{\bf u}_{j_2}={\bf u}'_{i}} \min \{A_{j_1}[{\PP}_i, {\bf u}_{j_1}], A_{j_2}[{\PP}_i, {\bf u}_{j_2}]\},$$ 
where each element ${\bf u}'_i(v)$ of ${\bf u}'_i$ is defined as ${\bf u}'_i(v)={\bf u}_i(v)-\sum_{u\in N(v)\cap C_v}w_{vu}$. 

Since the size of a DP table of each bag is $(\tw W)^{O(\tw)}$ and  each recursive formula can be computed in time $(\tw W)^{O(\tw)}$, the total running time is $(\tw W)^{O(\tw)}n$.
\end{proof}

Theorem \ref{thm:pseudo_tw:Egal} implies that if $W$ is bounded by a polynomial in $n$,  \textsc{Max Egalitarian}  can be computed in time $n^{O(\tw)}$.

\bibliographystyle{plainurl}
\bibliography{ref}

\end{document}